\documentclass{article}
\pdfoutput=1
\usepackage{graphicx,amsthm,amssymb,amsmath}

\theoremstyle{plain}
\newtheorem{thm}{Theorem}
\newtheorem{lem}{Lemma}

\newtheorem{prb}{Problem}

\theoremstyle{definition}

\theoremstyle{remark}

\title{Minimizing the Maximum Interference is Hard}%
\author{Kevin Buchin\thanks{Department of Mathematics and Computer Science, University of Technology Eindhoven, {\tt kbuchin@tue.nl}. I would like to thank Roger Wattenhofer and Yuezhou Lv for helpful comments.}}%
\date{}

\begin{document}

\maketitle

\begin{abstract}
We consider the following interference model for wireless sensor and ad hoc networks:
the receiver interference of a node is the number of transmission ranges it lies in.
We model transmission ranges as disks. For this case we show that
choosing transmission radii which
minimize the maximum interference while maintaining a connected symmetric
communication graph is NP-complete.
\end{abstract}

\section{Introduction}
Limiting the interference between nodes in a sensor network is substantial
for the energy-efficiency of the network. A common approach to reduce
interference is \emph{topology control}, i.e., restricting the communication
graph (see~\cite{bb-tc-07,k-isnr-07}). A theoretical problem in topology control which has been stated as
essential to understanding sensor networks is the following.

\begin{prb}[Locher, von Rickenbach, Wattenhofer~\cite{lrw-sncp-08}]\label{prb:tc}
Given $n$ nodes in the plane. Connect the nodes by a spanning tree.
For each node $v$ we construct a disk centering at $v$ with radius equal
to the distance to $v$'s furthest neighbor in the spanning tree. The interference
of a node $v$ is then defined as the number of disks that include node $v$
(not counting the disk of $v$ itself). Find
a spanning tree that minimizes the maximum interference.
\end{prb}
The choice of the radii as given in the problem statement guarantees that
the symmetric communication graph contains a spanning tree, i.e., that
the symmetric communication graph is connected.
The \emph{symmetric communication graph} is the undirected graph on the nodes
with edges between nodes which both lie in each others
\emph{transmission ranges}, i.e., in each others circles. We refer to the
radii of the circles as \emph{transmission radii}.

We prove that Problem~\ref{prb:tc} is NP-hard.
So far no lower bounds for the problem were known. Halld{\'o}rsson and
Tokuyama~\cite{ht-miwan-06} give an algorithm which yields a
maximum interference in $O(\sqrt{n})$.
An open problem that remains is to narrow the gap between this upper bound and
our lower bound.
For the case of points on a line there is a
$\sqrt[4]{n}$-approximation algorithm~\cite{rswz-rim-05}.
In a generalized version of the problem there is a positive real value
associated with each (ordered) pair of nodes, and the first node can send a message
to the second node (but will also interfere with it) if its transmission power is above this
value. In this version an approximation within less than a logarithmic factor in polynomial
time is not possible unless NP has slightly superpolynomial time algorithms~\cite{bp-cmi-06}.

\section{NP-Completeness}
In this section we prove that deciding whether the maximum interference
of a network is at most $3$ is NP-complete. Strictly speaking, this
implies that the interference in Problem~\ref{prb:tc} cannot be
approximated within a factor less than $4/3$ efficiently, since
it is not possible to distinguish between interference $3$ and $4$ in
polynomial time unless $P=NP$.

We prove the NP-hardness by a polynomial reduction from the problem of finding a \emph{Hamilton path} in a
\emph{grid graph} of maximum degree $3$. A (vertex-induced) grid graph is a graph for which the vertex set is a finite subset of
the two-dimensional integer grid $\mathbb{Z} \times \mathbb{Z}$ and there is an edge between two
vertices $x,y$ exactly if $x,y$ are neighbors on the grid, i.e., $\|x-y\| = 1$. We identify the corresponding
edge with the line segment from $x$ to $y$.
A Hamilton path in
a graph is a path in the graph with every vertex lying exactly once on the path. Deciding whether a Hamilton
path in a grid graph with maximum degree $3$ exists is NP-hard~\cite{pv-tgp-84}.

For the reduction we need for any grid graph of maximum degree $3$ a polynomial construction of a set of nodes such that
there is a Hamilton path in the grid graph exactly if there is a spanning tree with maximum interference at most $3$.
We may assume that the grid graph has no isolated vertex because in that case there is no Hamilton path and
we can check this in linear time.

A vertex $x\in\mathbb{Z}\times\mathbb{Z}$ of the grid graph is represented by a set of nodes
(which we call a \emph{vertex gadget}) containing the following nodes:
\begin{itemize}
\item a \emph{center node}: a node at position $x$,
\item \emph{satellite nodes}: three further nodes at three disjoint positions from the set $\{x\pm (0,1/4),\ x\pm (1/4,0)\}$. The
satellites are chosen such that the vertex gadget has a satellite node on every edge at $x$ of the
grid graph. For every vertex of degree less than 3, the remaining satellite is placed such that the distance between the remaining satellites of neighboring degree-2 vertices is larger than 1. This can be achieved by greedily placing the remaining satellites for each chain of degree-2 vertices.
\end{itemize}
Two satellites (from different vertex gadgets) on the same edge of the grid graph are called
\emph{partners}. Figure~\ref{fig:gadget} shows a grid graph of maximum degree $3$ and a
corresponding node set.
\begin{figure}
  \centering
  \includegraphics[width=0.8\textwidth]{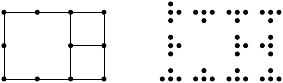}
  \caption{Example of a grid graph and a corresponding set of nodes.}
  \label{fig:gadget}
\end{figure}

For the NP-hardness we need to prove that the grid graph has a Hamilton path exactly if
the corresponding node set has a spanning tree with interference at most $3$.
We get one of the implications by constructing such a tree from a Hamilton path.

\begin{lem}\label{lem:HPtoST}
If a grid graph has a Hamilton path then the corresponding set of nodes
has a spanning tree with interference $3$.
\end{lem}
\begin{proof}
Given a grid graph with Hamilton path we can construct a spanning tree
with interference $3$ in the following way: For an arbitrary Hamilton path
\begin{itemize}
\item connect each center node to its satellites,
\item connect satellite partners if they lie on an edge of the Hamilton
cycle.
\end{itemize}
Center nodes and satellite nodes without a partner have transmission radius
$1/4$, while satellites with partners have transmission radius $1/2$.

This yields the following interferences:
A center node is in the transmission range of its satellites. It
is not in the transmission range of any other node since it has distance
at least $3/4$ to any other node. Thus the interference at a center node is
$3$.

A satellite is in the transmission range of the center node.
It is in the transmission range of any (other) satellite in its vertex gadget that
connects to a partner. If it connects to its partner, it is in the
transmission range of the partner. There can be no further interference
at a satellite since all other nodes have distance at least $3/4$ to the
satellite.

In a (Hamilton) path every vertex of the
grid graph is connected to at most two other vertices. Therefore in a
vertex gadget at most two satellites connect to their partners. This
yields an interference of at most 3 at satellites.
\end{proof}

Next we show that if the
interference induced by a spanning is at most 3
then in the spanning tree
vertex gadgets may only connect through partners.
 \begin{lem}\label{lem:partners}
Assume a grid graph has no isolated vertices. If a spanning tree on the corresponding set
of nodes has an edge between two different vertex gadgets other than an edge between
partners then there is a node with interference at least $4$.
\end{lem}
\begin{proof}
Suppose a satellite connects to a node which is further away than its partner. In this
case it contains at least one center node outside of its vertex gadget in its
transmission range.
Since this center node will also lie in the transmission ranges of its satellites, it
will have interference at least $4$.
\end{proof}

\begin{lem}\label{lem:STtoHP}
If the node set corresponding to a grid graph without isolated vertices has
a spanning tree with interference at most $3$ then the grid graph has a
Hamilton path.
\end{lem}
\begin{proof}
Suppose we have a spanning tree in which from each vertex gadget at most two satellites
connect to partners. Then this directly gives us a Hamilton path in the corresponding
grid graph by simply connecting the vertices in the same way as the vertex gadgets.

Now assume there is a spanning tree with interference at most $3$ which is not
of this type. The spanning tree has a vertex gadget that connects to at least three
other vertex gadgets and by Lemma~\ref{lem:partners} these must be connections from
satellites to their partners. Thus all three satellites in the vertex gadget connect to their
partners. Now all three satellites lie in the transmission range of their partner, of the
other two satellites, and of the center node of the gadget. Therefore, the satellites have interference
at least $4$ contradicting the assumption of interference $3$.
\end{proof}

\begin{thm}\label{thm:nphard}
Deciding whether a set of nodes in the plane has a spanning tree with
interference at most $3$ is NP-complete.
\end{thm}
\begin{proof}
The polynomial construction of the node set from the grid graph together
with Lemmas~\ref{lem:HPtoST} and~\ref{lem:STtoHP} directly yield the NP-hardness.

To verify whether a spanning tree has a certain interference
it suffices to perform
$n \choose 2$ in-circle tests. Thus, the problem is in NP.
\end{proof}



{\small
\begin{thebibliography}{1}

\bibitem{bp-cmi-06}
D.~Bil{\`o} and G.~Proietti.
\newblock On the complexity of minimizing interference in ad-hoc and sensor
  networks.
\newblock In {\em Proc.\ 2nd Internat.\ Workshop Algorithmic Aspects of
  Wireless Sensor Networks (ALGOSENSOR)}, vol.\ 4240 of {\em LNCS}, pp.\  13--24, 2006.

\bibitem{bb-tc-07}
K.~Buchin and M.~Buchin.
\newblock Topology control.
\newblock In D.~Wagner and R.~Wattenhofer, editors, {\em Algorithms for Sensor
  and Ad Hoc Networks}, vol.\ 4621 of {\em LNCS}, pp.\ 81--98. Springer,
  2007.

\bibitem{ht-miwan-06}
M.~M. Halld{\'o}rsson and T.~Tokuyama.
\newblock Minimizing interference of a wireless ad-hoc network in a plane.
\newblock In {\em Proc.\ 2nd Internat.\ Workshop Algorithmic Aspects of
  Wireless Sensor Networks (ALGOSENSOR)}, vol.\ 4240 of {\em LNCS}, pp.\
  71--82, 2006.

\bibitem{k-isnr-07}
A.~Kr{\"o}ller.
\newblock Interference and signal-to-noise-ratio.
\newblock In D.~Wagner and R.~Wattenhofer, editors, {\em Algorithms for Sensor
  and Ad Hoc Networks}, vol.\ 4621 of {\em LNCS}, pp.\ 99--116. Springer,
  2007.

\bibitem{lrw-sncp-08}
T.~Locher, P.~von Rickenbach, and R.~Wattenhofer.
\newblock Sensor networks continue to puzzle: Selected open problems.
\newblock In {\em Proc.\ 9th Internat.\ Conf.\ Distributed Computing and
  Networking (ICDCN)}, 2008.

\bibitem{pv-tgp-84}
C.~H. Papadimitriou and U.~V. Vazirani.
\newblock On two geometric problems related to the traveling salesman problem.
\newblock {\em J. Algorithms}, 5(2):231--246, 1984.

\bibitem{rswz-rim-05}
P.~von Rickenbach, S.~Schmid, R.~Wattenhofer, and A.~Zollinger.
\newblock A robust interference model for wireless ad-hoc networks.
\newblock In {\em Proc.\ 5th Internat. Workshop Algorithms for Wireless,
  Mobile, Ad Hoc and Sensor Networks (WMAN)}, 2005.

\end{thebibliography}
}
\end{document}